\documentclass[10pt,twocolumn,twoside]{IEEEtran}
\usepackage{subfigure}
\usepackage{setspace}
\usepackage{amsmath}
\usepackage{amssymb}
\usepackage{amsfonts}
\usepackage{amscd}
\usepackage{mathdots}
\usepackage{mathrsfs}
\usepackage[final]{graphicx}
\usepackage{graphicx}
\usepackage{psfrag}
\usepackage{epsfig}
\usepackage{color}
\usepackage{url}
\usepackage{textcomp}
\usepackage{multirow}
\input{epsf.sty}
\newtheorem{theorem}{Theorem}
\newtheorem{lemma}{Lemma}
\newtheorem{definition}{Definition}

\begin{document}
\markboth{IEEE Transactions on Information Theory, Vol. XX, No. XX, XXXX}{Sanjay \MakeLowercase{\textit{et al.}}: Maximum Rate of Unitary-Weight, Single-Symbol Decodable STBCs}

\title{Maximum Rate of Unitary-Weight, Single-Symbol Decodable STBCs}

\vspace{1.00cm}

\author{Sanjay Karmakar, K. Pavan Srinath and B. Sundar Rajan, Senior Member, IEEE
\thanks{This work was partly supported by the DRDO-IISc Program on Advanced Research in Mathematical
Engineering and by the Council of Scientific \&
Industrial Research (CSIR), India, through Research Grant (22(0365)/04/EMR-II) to B.S.~Rajan.}
\thanks{Sanjay Karmakar is with the Electrical and Computer Science Department at the University of Colorado at Boulder. Part of this work was carried out when he was at the Indian Institute of Science, Bangalore}
 \thanks{K. Pavan Srinath and B. Sundar Rajan are with the Department of Electrical Communication Engineering, Indian Institute of Science, Bangalore - 560012. email:bsrajan@ece.iisc.ernet.in.}
\thanks{Different parts of the content of this paper appeared in the Proc. of IEEE International Symposium on Information on Information Theory (ISIT 2006), Seattle, Washington, July 09-14, 2006.}
}

\markboth{IEEE Transactions on Information Theory, Vol. XX, No. XX, XXXX}{Sanjay \MakeLowercase{\textit{et al.}}: Maximum Rate of Unitary-Weight, Single-Symbol Decodable STBCs}

%\date{\today}
%\begin{document}
\maketitle

\begin{abstract}
It is well known that the Space-time Block Codes (STBCs) from Complex orthogonal designs (CODs) are  single-symbol decodable/symbol-by-symbol decodable (SSD). The weight matrices of the square CODs are all unitary and obtainable from the unitary matrix representations of Clifford Algebras when the number of transmit antennas $n$ is a power of 2. The rate of the square CODs for $n = 2^a$ has been shown to be $\frac{a+1}{2^a}$ complex symbols per channel use. However, SSD codes having unitary-weight matrices need not be CODs, an example being the Minimum-Decoding-Complexity STBCs from Quasi-Orthogonal Designs. In this paper, an achievable upper bound on the rate of any unitary-weight SSD code is derived to be $\frac{a}{2^{a-1}}$ complex symbols per channel use for $2^a$ antennas, and this upper bound is larger than that of the CODs. By way of code construction, the interrelationship between the weight matrices of unitary-weight SSD codes is studied. Also, the coding gain of all unitary-weight SSD codes is proved to be the same for QAM constellations and conditions that are necessary for unitary-weight SSD codes to achieve full transmit diversity and optimum coding gain are presented.
\end{abstract}


%%%%%%%%% Section I begins %%%%%%%%%%%%%%%%%

\section{Introduction}
\label{sec1}
\pagenumbering{arabic}
Space-Time Block Codes from Complex Orthogonal Designs \cite{TJC} are popular because they offer full transmit diversity for any arbitrary signal constellation and also are single-symbol decodable. In fact, CODs are \emph{single-real-symbol} decodable for constellations such as the rectangular QAM, which can be expressed as a Cartesian product of two PAM constellations, while for constellations such as PSK, CODs are \emph{single-complex-symbol} decodable. The \emph{weight matrices}, also called \emph{linear dispersion matrices} \cite{HaH} (refer Subsection \ref{ssd_codes} for a definition of weight matrices), of the square CODs are all unitary, and a detailed construction method to obtain these weight matrices from irreducible matrix representations of Clifford Algebras has been presented in \cite{TiH} for $2^a$ transmit antennas. It has also been shown that the maximum rate of the square CODs for $2^a$ transmit antennas is $\frac{a+1}{2^a}$ complex symbols per channel use. Although rectangular CODs \cite{liang} offer a higher rate, they are not delay efficient, making square CODs more attractive in practice.

In general, single-complex-symbol decodable codes need not be CODs. Throughout this paper, unless otherwise mentioned, SSD codes refer to single-complex-symbol decodable codes. The Co-ordinate Interleaved Orthogonal Designs (CIODs) \cite{KhR} have been shown to be SSD codes, while offering full transmit diversity for specific complex constellations only. However, the CIODs have non-unitary-weight matrices. SSD codes, that include unitary-weight codes and rectangular designs, have been reported in \cite{WWX}, \cite{YGT} and are popularly known as Minimum-Decoding-Complexity codes from Quasi-orthogonal designs (MDCQODs). The rates of both the CIODs and the class of codes reported in \cite{YGT} for $2^a$ transmit antennas have been shown to be $\frac{a}{2^{a-1}}$ complex symbols per channel use. In \cite{YGT}, the maximum rate of the MDCQODs has been reported, and this rate includes that for rectangular designs. However, the maximum rate of general SSD codes has not been reported so far in the literature, to the best of our knowledge.

In this paper, we make the following contributions.
\begin{itemize}
\item We derive an upper bound on the rate of \textit{unitary-weight SSD codes}  for $2^a$ transmit antennas. This upper bound is found to be $\frac{a}{2^{a-1}}$ complex symbols per channel use.
\item We give a general construction method to obtain codes that meet this upper bound and further show the interrelationship between the weight matrices of general unitary-weight SSD codes. All known unitary-weight SSD codes including square MDCQODs are special cases of this construction.
\item We prove that all unitary-weight SSD codes have the same coding gain and specifically for QAM constellations, we provide the angle of rotation that ensures full transmit diversity and optimum coding gain for all unitary-weight SSD codes.
\end{itemize}

The organization of the paper is as follows. Section \ref{sys_model} gives the system model and relevant definitions. Section \ref{sec3} introduces the notion of normalization and its use in the analysis of unitary-weight SSD codes. Section \ref{sec4} provides the upper bound on the rate of unitary-weight SSD codes and the structure of general unitary-weight SSD codes. Diversity conditions, coding gain calculations for QAM and simulation results are given in Section \ref{sec5}, Subsections \ref{sec5a} and \ref{sec5b}, respectively. Discussions on the direction for future research constitute Section \ref{sec6}.

\textit{\textit{Notations}}: $\mathbb{R}$ and $\mathbb{C}$ denote the field of real and complex numbers, respectively and $j$ represents $\sqrt{-1}$. $GL(n,\mathbb{C})$ denotes the group of invertible matrices of size $n \times n$ with complex entries. For any complex matrix $A$, $tr(A)$, $\Vert A \Vert$, $A^H$ and $det(A)$ represent the trace, the Frobenius norm, the Hermitian  and  the determinant of $A$, respectively. $I_n$ and $O_n$ represent the $n \times n$ identity matrix and the zero matrix, respectively. For a complex random variable $X$, $X \sim \mathcal{N}_{\mathbb{C}}(0,N)$ denotes that $X$ has a complex normal distribution with mean $0$ and variance $N$. For a complex variable $x$, $x_I$ and $x_Q$ represent the real and the imaginary parts of $x$, respectively, and $\vert x \vert$ denotes the absolute value of $x$. For a set $\mathcal{S}$, $\vert \mathcal{S} \vert$ denotes the cardinality of $\mathcal{S}$.

\section{System Model}
\label{sys_model}
We consider Rayleigh quasi-static flat-fading MIMO channel with full channel state information (CSI) at the receiver and no CSI at the transmitter. We assume a MIMO system with $n$ transmit antennas and $m$ receive antennas. Since we are considering only square STBCs in this paper, the number of time slots is also $n$. The channel model is
\begin{equation*}
Y = SH + N,
\end{equation*}
\noindent where $S \in \mathbb{C}^{n\times n}$ is the codeword matrix, transmitted over $n$ channel uses, $N \in \mathbb{C}^{n \times m}$ is a complex white Gaussian noise matrix with i.i.d. entries $\sim
\mathcal{N}_{\mathbb{C}}\left(0,N_{0}\right)$, $H \in \mathbb{C}^{n \times m}$ is the channel matrix with the entries assumed to be i.i.d. circularly symmetric Gaussian random variables $\sim \mathcal{N}_\mathbb{C}\left(0,1\right)$ and $Y \in \mathbb{C}^{n \times m}$ is the received matrix.

\begin{definition}\label{def}(\textit{STBC}) A \textit{space-time block code} $\mathcal{S}$ is a set of complex matrices called codeword matrices. For a system with $n$ transmit antennas, a codeword matrix is a $T \times n$ matrix, where $T$ is the number of time slots ($T = n$ in this paper) and the $(i,j)^{th}$ entry of the codeword matrix refers to the signal transmitted by the $j^{th}$ transmit antenna in the $i^{th}$ time slot.
\end{definition}

\begin{definition}\label{def1}$\left(\textit{Code rate}\right)$ If there are $k$ independent complex information symbols in the codeword which are transmitted over $T$ channel uses, then, the code rate is defined to be $k/T$  complex symbols per channel use. For instance, for the Alamouti code, $k=2$ and $T=2$. So, its code rate is 1 complex symbol per channel use.
\end{definition}

\begin{definition}\label{def2}(\textit{Full-Diversity Code}) An  STBC encoding symbols chosen from a constellation $\mathcal{A}$ is said to offer \emph{full-diversity} iff for every possible codeword pair $(S,\hat{S})$, with $S \neq \hat{S}$, the codeword difference matrix $S-\hat{S}$ is full-ranked \cite{TSC}.
\end{definition}

In general, whether a code offers full-diversity or not depends on the constellation that it employs. A code can offer full diversity for a certain complex constellation $\mathcal{A}$ but not for another complex constellation. The CODs are special in this aspect since they offer full-diversity for any arbitrary complex constellation.

\begin{definition}\label{def3}(\textit{Coding Gain}) The coding gain $\delta$ of an STBC is defined as 
\begin{equation*}
\delta = \min_{S-\hat{S},S \neq \hat{S}}\left(\prod_{i=1}^{r}\lambda_i\right)^{\frac{1}{r}},
\end{equation*}
\noindent where  $\lambda_i, i = 1,2,\cdots,r$, are the non-zero eigenvalues of the matrix $\left(S - \hat{S}\right)^H\left(S - \hat{S}\right)$ and $r$ is the minimum of the rank of $\left(S - \hat{S}\right)^H\left(S - \hat{S}\right)$ for all possible codeword pairs $(S, \hat{S})$, $S \neq \hat{S}$.
\end{definition}

If the code offers full-diversity for a constellation $\mathcal{A}$, then, the coding gain is  $\delta_{min}^{\frac{1}{n}}$, where $\delta_{min}$ is the minimum of the determinant of the matrix $\left(S - \hat{S}\right)^H\left(S - \hat{S}\right)$ among all possible codeword matrix pairs $(S, \hat{S})$, with $S \neq \hat{S}$.

\subsection{Single-Symbol Decodable Codes}\label{ssd_codes}

In this subsection, we formally define and classify linear SSD codes. Any $n\times n$ codeword matrix $S$ of a linear dispersion STBC $\mathcal{S}$ with $k$ complex information symbols $x_1,x_2,\cdots,x_k$ can be expressed as
\begin{equation}
\label{ldceqn}
S  = \sum_{i=1}^{k} (x_{iI}A_{iI}+x_{iQ}A_{iQ}),
\end{equation}
\noindent where $x_i= x_{iI}+jx_{iQ}$,  $1\leq i \leq k,$  take values from a complex constellation ${\cal A}$. Then, $\vert \mathcal{S} \vert$, i.e., the  number of codewords, is $\vert{\cal A}\vert^k$. The set of $n\times n$ complex matrices $\{A_{iI}, A_{iQ} \}, 1 \leq i \leq k$, called {\it weight matrices} define $\mathcal{S}$. Notice that in \eqref{ldceqn}, all the $2k$ weight matrices are required to form a linearly independent set over $\mathbb{R}$, since we are transmitting $k$ independent information symbols.
%%%%%%%%%%%%%%%%%%%%%%%%%%%%%%%%%%%%%%%%%%%%%%%%%%

Assuming that perfect channel state information (CSI) is available at the receiver, the maximum likelihood (ML) decision rule  minimizes the metric,
\begin{equation}
\label{mlmetric}
M(S) \triangleq tr({({Y-SH})}^{H}({Y-SH})) = {\Vert {Y}-{SH} \Vert}^2.
\end{equation}
Since there are $\vert \mathcal{A} \vert ^k$ different codewords, in general, ML decoding requires $\vert \mathcal{A} \vert ^k$ computations, one for each codeword. Suppose the set of weight matrices are chosen such that the decoding metric \eqref{mlmetric} could be decomposed as 
\begin{equation*}
M(S)  = \sum_{j=1}^p {f_j(x_{(j-1)q+1},x_{(j-1)q+2},\cdots, x_{(j-1)q+q})},
\end{equation*}
which is a sum of $p$ positive terms, each involving exactly $q$ complex variables only, where $pq=k$. Then, decoding requires $\sum_{j=1}^p |\mathcal{A}|^q = p |\mathcal{A}|^q$ computations and the code is called a $q$-symbol decodable code \cite{sanjay}. The case $q=1$ corresponds to $SSD$ codes that include the well known CODs as a proper subclass, and have been extensively studied \cite{TJC}, \cite{TiH}, \cite{KhR}, \cite{WWX}, \cite{YGT}, \cite{Ala}. The codes corresponding to $q=2$, are called \textit{Double-Symbol-Decodable (DSD)} codes. The Quasi-Orthogonal Designs studied in \cite{Jaf}, \cite{SuX}, \cite{ShP} are proper subclasses of DSD codes.
\begin{definition}\cite{TJC}
A square {\textit{complex orthogonal design}} $\mathcal{S}$ for $n$ transmit antennas is a set of codeword matrices of size $ n \times n $, with each codeword matrix $S$ satisfying the following conditions:
\begin{itemize}
\item the entries of $S$ are complex linear combination of  $ x_1 , x_2 , \cdots, x_k $ and their complex conjugates $ x_1^* , x_2^* , \cdots, x_k^* $.
\item (Orthonormality:)
\begin{eqnarray*}
  S^H S=({\vert x_1 \vert}^2 + \cdots+{\vert x_k \vert}^2)I_n
 \end{eqnarray*}
holds for any complex values for $ x_i , i = 1, 2,\cdots, k$.
\end{itemize}
\end{definition}

 A set of necessary and sufficient conditions for $\mathcal{S}$ to be a COD is \cite{TJC,TiH}
\begin{equation}
\label{A1}
A_{iI}^{H}A_{iI}  =A_{iQ}^{H}A_{iQ}  =I_n,~~~i=1,2,\cdots,k;
\end{equation}
\vspace{-7mm}
\begin{subequations}
\label{A2}
\begin{align}
A_{iI}^{H}A_{jQ} + A_{jQ}^{H}A_{iI} & =O_n \label{A21}, \\
A_{iI}^{H}A_{jI} + A_{jI}^{H}A_{iI} & =O_n \label{A22}, \\
A_{iQ}^{H}A_{jQ} + A_{jQ}^{H}A_{iQ} & =O_n\label{A23},
\end{align}
\end{subequations}
for $1 \leq i \neq j \leq k $, and
\begin{eqnarray}
\label{A3}
A_{iI}^{H}A_{iQ}+ A_{iQ}^{H}A_{iI}  =O_n,  ~~~   i=1,2,\cdots,k.
\end{eqnarray}
\noindent

STBCs obtained from CODs \cite{TJC}, \cite{TiH} are SSD like the well known Alamouti code \cite{Ala}, and  satisfy \eqref{A1}, \eqref{A2} and \eqref{A3}. For $S$ to be  SSD, it is not necessary that it satisfies \eqref{A1} and \eqref{A3}, i.e., it is sufficient that it satisfies only  \eqref{A2} - this result was shown in \cite{KhR}. Since then, different classes of SSD codes that are not CODs have been studied by several authors, \cite{KhR}, \cite{WWX}, \cite{YGT}. SSD codes can be systematically classified as follows.
\begin{enumerate}
\item Linear STBCs satisfying \eqref{A1}, \eqref{A2} and \eqref{A3} are \textit{CODs}.
\item Linear STBCs satisfying \eqref{A2} are called {\it SSD codes}. These may or may not satisfy \eqref{A1} and \eqref{A3}.
\item Linear STBCs satisfying \eqref{A1} and \eqref{A2} and not satisfying \eqref{A3} are called \textit{Unitary-Weight SSD codes}.
\item Linear STBCs satisfying \eqref{A2} and not satisfying \eqref{A1} are called \textit{Non-Unitary-weight SSD codes}. These may or may not satisfy \eqref{A3}.
\end{enumerate}
The codes discussed in  \cite{KhR}, which are called CIODs, constitute an example class of Non-unitary-weight SSD codes. The classes of codes  studied in \cite{YGT} are unitary-weight SSD codes. The classes of codes studied in \cite{WWX}, called Minimum Decoding Complexity codes from Quasi-Orthogonal Designs (MDCQOD codes), include some unitary-weight SSD codes as well as non unitary-weight SSD codes.
%%%%%%%%%%%%%%%%%%%%%%%%%%%%%%%%%%%%

The notion of SSD codes have been extended to coding for MIMO-OFDM systems in \cite{DSKR,GoR} and recently, low-decoding complexity codes called 2-group and 4-group decodable codes \cite{KiR1,KiR2,KiR3,RaR} and SSD codes \cite{YiK} in particular have been studied for use in cooperative networks as distributed STBCs.

%%%%%%%%%%%%%%%%%%%%%% Section II begins %%%%%%%%%%%%%%%%%%%
\section{Unitary-Weight SSD codes}
\label{sec3}
In this section, we analyze the structure of the weight matrices of unitary-weight SSD codes. We make use of the following lemma in our analysis.
\begin{lemma}
\label{unitary_mult}
Let $\mathcal{S} = \{ S | S=\sum_{i=1}^{k}x_{iI}A_{iI}+x_{iQ}A_{iQ} , x_i \in \mathcal{A}\}$ be a unitary-weight STBC and consider the STBC $\mathcal{S}_U \triangleq  \{ US  | S \in \mathcal{S}\}$, where $U$ is any unitary matrix.  Then, $\mathcal{S}_U$ is SSD iff $\mathcal{S}$ is SSD. Further, both the codes have the same coding gain for the constellation $\mathcal{A}$.
\end{lemma}

\begin{proof}
The proof is straightforward. For the STBC $\mathcal{S}_U$, the weight matrices are $UA_{1I}$, $UA_{iQ}$, $i = 1,2,\cdots,k$. It is easy to verify that if the matrices $A_{iI}$, $A_{iQ}$, $i = 1,2,\cdots,k$ satisfy \eqref{A2}, then, the matrices  $UA_{1I}$, $UA_{iQ}$, $i = 1,2,\cdots,k$ also satisfy \eqref{A2} and vice-versa. Further, for any pair of distinct codeword matrices $S$ and $\hat{S}$, the eigenvalues of $(S-\hat{S})^H(S-\hat{s})$ are the same as that of $(US-U\hat{S})^H(US-U\hat{S})$, making the coding gain the same for both the STBCs.
\end{proof}
The STBCs $\mathcal{S}$ and $\mathcal{S}_U$ are said to be \emph{equivalent}. To simplify our analysis of unitary-weight codes, we make use of \emph{normalization} as described below. Let $\mathcal{S}$ be a unitary-weight STBC and let its codeword matrix $S$ be expressed as
\begin{equation*}
S = \sum_{i=1}^{k} (x_{iI}A_{iI}^\prime + x_{iQ}A_{iQ}^\prime).
\end{equation*}
Consider the code $\mathcal{S}_N \triangleq \{A_{1I}^{\prime H}S | S \in \mathcal{S}\}$. Clearly, from Lemma \ref{unitary_mult}, $\mathcal{S}_N$ is equivalent to $\mathcal{S}$. The weight matrices of $\mathcal{S}_N $ are
\begin{eqnarray*}
\label{normalizessd}
\begin{array}{rl}
A_{iI} & = A_{1I}^{\prime H} A_{iI}^\prime, \\
A_{iQ} & = A_{1I}^{\prime H} A_{iQ}^\prime.
\end{array}
\end{eqnarray*}
With this, a codeword matrix $S_N$ of $\mathcal{S}_N$ can be written as
\begin{equation*}
S_N  = x_{1I}I_n+x_{1Q}A_{1Q}+\sum_{i=2}^{k} (x_{iI}A_{iI}+x_{iQ}A_{iQ}).
\end{equation*}
We call the code $\mathcal{S}_N $ to be the normalized code of $\mathcal{S}$. In general, any unitary-weight SSD code with one of its weight matrices being the identity matrix is called \emph{normalized} unitary-weight SSD code. Studying unitary-weight SSD codes becomes simpler by studying normalized unitary-weight SSD codes. Also, an upper bound on the rate of unitary-weight SSD codes is the same as that of the normalized unitary-weight SSD codes. For the normalized unitary-weight SSD code transmitting $k$ symbols in $n$ channel uses, the conditions presented in \eqref{A1} and \eqref{A2} can be rewritten as 
\begin{eqnarray}
A_{iI}^H &=&-A_{iI}~(equivalently,~A_{iI}^2 = -I_n),\label{B11} \\
A_{iQ}^H&=&-A_{iQ}~(equivalently,~A_{iQ}^2 = -I_n),\label{B12} \\
A_{1Q}^HA_{iI} &=& A_{iI}A_{1Q}, \label{B13} \\
A_{1Q}^HA_{iQ} &=& A_{iQ}A_{1Q}, \label{B14}
\end{eqnarray}
for $i = 2,\cdots,k$, and
\begin{eqnarray}
A_{iI}A_{jI} &=& -A_{jI}A_{iI} \label{B15}, \\
A_{iQ}A_{jQ} &=& -A_{jQ}A_{iQ} \label{B16}, \\
A_{iI}A_{jQ} &=& -A_{jQ}A_{iI} \label{B17},
\end{eqnarray}
for $2 \leq i \neq j \leq k $. So, every weight matrix except $A_{1I} = I_n$ and $A_{1Q}$ should square to $-I_n$. Shown below is the grouping of weight matrices (We will later show that $A_{iQ} = \pm A_{1Q}A_{iI}$, $i =2, \cdots k$).
\begin{center}
\begin{tabular}{l|l|l|l}
$A_{1I} = I_n$  & $A_{2I}$ & $\cdots$ & $A_{kI}$ \\ \hline
$A_{1Q}$  & $A_{2Q}$ & $\cdots$ & $A_{kQ}$ \\
\end{tabular}
\end{center}
Except $I_n$, the elements in the first row should mutually anticommute and also square to $-I_n$. From \eqref{B13} and \eqref{B14}, it is clear that if $A_{1Q}^2 = -I_n$, then, $A_{1Q}$ should anticommute with all the weight matrices except $A_{1I} = I_n$.  So, the upper bound on the rate of a unitary-weight SSD code is determined by the number of mutually anticommuting unitary matrices. The following section deals with determining the upper bound.

%%%%%%%%%%%%%%%%%%%%%%%%%%%%%%%%%%%%%%%%%%%%%%%%%%%%%%%%%%%%%%%%%%%%%%%%%%%

\section{An upper bound on the rate of Unitary-Weight SSD codes}
\label{sec4}
In this section, we determine the upper bound on the rate of unitary-weight SSD codes and also give a general construction scheme to obtain codes meeting the upper bound. To do so, we make use of the following lemmas regarding matrices of size $n\times n$.

\begin{lemma}\label{lemma_2}
\cite{anti_matric} Consider $n \times n$ matrices with complex entries.
\begin{enumerate}
\item If $n = 2^an_0$, with $n_0$ odd, then there are $l$ elements of $GL(n,\mathbb{C})$ that anticommute pairwise if and only if $l \leq 2a+1$.	
\item If $n=2^a$ and matrices $F_1, \cdots, F_{2a}$ anticommute pairwise, then the set of products $F_{i_1}F_{i_2}\cdots F_{i_s}$ with $1 \leq i_1 < \cdots < i_s \leq 2a$ along with $I_n$ forms a basis for the $2^{2a}$ dimensional space of all $n \times n$ matrices over $\mathbb{C}$. In each case $F_{i}^2$ is a scaled identity matrix.
\end{enumerate}
\end{lemma}
\begin{proof}
Available in \cite{anti_matric}.
\end{proof}

Let $F_1, \cdots, F_{2a}$ be anticommuting, anti-Hermitian, unitary matrices (so that $F_i^2 = -I_n$, $i = 1, 2,\cdots, 2a$). The following two lemmas are applicable for such matrices.

\begin{lemma}\label{lemma_3}
The product $F_{i_1}F_{i_2}\cdots F_{i_s}$ with $1 \leq i_1 < \cdots < i_s \leq 2a$ squares to $(-1)^{\frac{s(s+1)}{2}}I_n$.
\end{lemma}
\begin{proof}
$(F_{i_1}F_{i_2} \cdots F_{i_s})(F_{i_1}F_{i_2} \cdots F_{i_s})$
\begin{eqnarray*}
 & = & (-1)^{s-1}(F_{i_1}^2 F_{i_2} \cdots F_{i_s})(F_{i_2}F_{i_3} \cdots F_{i_s}) \\
& = & (-1)^{s-1}(-1)^{s-2}(F_{i_1}^2 F_{i_2}^2 \cdots F_{i_s})(F_{i_3}F_{i_4} \cdots F_{i_s}) \\
& = & (-1)^{[(s-1)+(s-2)+\cdots 1]}(F_{i_1}^2 F_{i_2}^2 \cdots F_{i_s}^2) \\
& = & (-1)^{\frac{s(s-1)}{2}}(-1)^{s}I_n \\
& = & (-1)^{\frac{s(s+1)}{2}}I_n.
\end{eqnarray*}
This proves the lemma.
\end{proof}

\begin{lemma}\label{lemma_4}
 Let $\Omega_1 = \{ F_{i_1},F_{i_2},\cdots,F_{i_s} \}$ and $\Omega_2$ $= \{ F_{j_1},F_{j_2},\cdots,F_{j_r} \}$
with $1 \leq i_1 < \cdots < i_s \leq 2a$ and $1 \leq j_1 < \cdots < j_r \leq 2a$. Let $\vert \Omega_1 \cap \Omega_2 \vert = p$. Then the product matrix $F_{i_1}F_{i_2}\cdots F_{i_s}$ commutes with $F_{j_1}F_{j_2}\cdots F_{j_r}$ if exactly one of the following is satisfied, and anticommutes otherwise.
\begin{enumerate}
\item $r,s$ and $p$ are all odd.
\item The product $rs$ is even and $p$ is even (including 0).
\end{enumerate}
\end{lemma}
\begin{proof}
If $F_{j_k} \in \Omega_1 \cap \Omega_2$, we note that
\begin{equation*}
(F_{i_1}F_{i_2} \cdots F_{i_s})F_{j_k} = (-1)^{s-1}F_{j_k}(F_{i_1}F_{i_2} \cdots F_{i_s}),
\end{equation*}
\noindent and if $F_{j_k} \notin \Omega_1 \cap \Omega_2$,
\begin{equation*}
(F_{i_1}F_{i_2} \cdots F_{i_s})F_{j_k} = (-1)^{s}F_{j_k}(F_{i_1}F_{i_2} \cdots F_{i_s}).
\end{equation*}
\noindent Now,\\
$(F_{i_1}F_{i_2} \cdots F_{i_s})(F_{j_1}F_{j_2} \cdots F_{j_r})$
\begin{eqnarray*}
& = & (-1)^{p(s-1)}(-1)^{(r-p)s}(F_{j_1}F_{j_2} \cdots F_{j_r})(F_{i_1}F_{i_2} \cdots F_{i_s})\\
& = & (-1)^{rs-p}(F_{j_1}F_{j_2} \cdots F_{j_r})(F_{i_1}F_{i_2} \cdots F_{i_s}).
\end{eqnarray*}
\emph{case} 1). Since $r,s$ and $p$ are all odd, $(-1)^{rs-p}$ = 1.\\
\emph{case} 2). The product $rs$ is even and  $p$ is even (including 0). So, $(-1)^{rs-p}$ = 1.
\end{proof}

From Lemma \ref{lemma_2}, the maximum number of pairwise anticommuting matrices of size  $2^a \times 2^a$ is $2a+1$. Hence, the maximum possible value of $k$, i.e., the number of complex information symbols, is $2a+2$, since we also consider $I_n$ as a weight matrix. In order to provide an achievable upper bound on the rate of unitary-weight SSD codes, we first assume that the case $k = 2a+2$ is a possibility. Denoting the $2a+1$ anticommuting matrices by $F_1$, $F_2, \cdots, F_{2a+1}$, we note from Lemma \ref{lemma_2} that the set $\{F_1^{\lambda_1}F_2^{\lambda_2}\cdots F_{2a}^{\lambda_{2a}}, \lambda_i \in \{0,1\}, i = 1, \cdots, 2a \}$ is a basis for $\mathbb{C}^{2^a \times 2^a}$ over $\mathbb{C}$. Therefore, $\{F_1^{\lambda_1}F_2^{\lambda_2}\cdots F_{2a}^{\lambda_{2a}}, jF_1^{\lambda_1}F_2^{\lambda_2}\cdots F_{2a}^{\lambda_{2a}}, \lambda_i \in \{0,1\}, i = 1, \cdots, 2a \}$ is a basis for $\mathbb{C}^{2^a \times 2^a}$ over $\mathbb{R}$. It can be checked by applying Lemma \ref{lemma_4} that the only product matrix that anticommutes with $F_1, F_2, \cdots,$ and $F_{2a}$ is $cF_1F_2\cdots F_{2a}$, where $c \in \mathbb{C}$.  So, it must be that $F_{2a+1} = cF_1F_2\cdots F_{2a}$, $c \in \mathbb{C}$.

For our construction, we need anticommuting, anti-Hermitian, unitary matrices (so that they square to $-I_n$). An excellent treatment of irreducible matrix representations of Clifford algebras is given in \cite{TiH} and the same paper also presents an algorithm to obtain $2a+1$ pairwise anticommuting $2^a \times 2^a$ matrices that all square to $-I_n$ ($n = 2^a$). In fact, these matrices are precisely the weight matrices (except $I_n$) of square CODs. As mentioned before, we denote them by $F_1$, $F_2, \cdots, F_{2a+1}$, with $F_{2a+1} = cF_1F_2\cdots F_{2a}$, and
\begin{equation*}
c = \left\{ \begin{array}{ccc}
 \pm j & \textrm{if} & (F_1F_2\cdots F_{2a})^2 = I_n, \\
 \pm 1 & \textrm{otherwise}. & \\
\end{array} \right.
\end{equation*}
It must be noted that the matrices obtained from \cite{TiH} are not unique, i.e., these are not the only set of mutually anticommuting, anti-Hermitian, unitary matrices of size $2^a \times 2^a$.
It can be noted by applying Lemma \ref{lemma_3} that $(F_1F_2\cdots F_{2a})^2 = -I_n$ when $a$ is odd. We are now ready to prove the main result of the paper.

\begin{theorem}
\label{main_thm}
The rate $\frac{k}{2^a}$ of a $2^a\times 2^a$ unitary-weight SSD code is upper bounded as 
\[
\frac{k}{2^a} \leq \frac{2a}{2^a}= \frac{a}{2^{a-1}}.
\]
\end{theorem}

\begin{proof}
We prove the theorem in three parts as follows. \\
\textit{Claim 1}: $k \neq 2a+2$. \\
To prove this, let us first suppose that $k = 2a+2$, in which case, we have the following grouping scheme.\\
%\begin{figure}
\begin{center}
\begin{tabular}{l|l|l|l|l}
$I_n$  & $F_1$ & $F_2$ & $\cdots$ & $F_{2a+1}$ \\ \hline
$A_{1Q}$  & $A_{2Q}$ & $A_{3Q}$ & $\cdots$ & $A_{(2a+2)Q}$ \\
\end{tabular}
\end{center}
%\label{fig_2}
%\end{figure}

Let $A_{iQ} = \sum_{j=1}^{2^{2a}}a_{i,j} F_1^{\lambda_{1,j}}F_2^{\lambda_{2,j}}\cdots F_{2a}^{\lambda_{2a,j}}$, $\lambda_{m,j} \in \{0,1\}$, $m = 1, 2, \cdots, 2a$, $i=1,2,\cdots,2a+2$ and $a_{i,j} \in \mathbb{C}$. This is possible because of Lemma \ref{lemma_2}. Considering $A_{2Q}$, since $A_{2Q}$ anticommutes with $F_2$, $F_3$, $\cdots$ and $F_{2a+1}$, every individual term of $A_{2Q}$ must anticommute with $F_2$, $F_3$, $\cdots$ and $F_{2a+1}$. So, we look for all possible candidates from the set $\{F_1^{\lambda_1}F_2^{\lambda_2}\cdots F_{2a}^{\lambda_{2a}}, \lambda_i \in \{0,1\}, i = 1, \cdots, 2a \}$ which anticommute with $F_2$, $F_3$, $\cdots$ and $F_{2a+1}$. By applying Lemma \ref{lemma_4}, the only possible choice is $F_1$. Since the weight matrices are required to be independent over $\mathbb{R}$ and in view of the condition in \eqref{B12}, there is no valid possibility for $A_{2Q}$. As a result, a unitary-weight $2^a \times 2^a$ SSD code with $2a+2$ independent complex symbols does not exist. \\

\noindent \textit{Claim 2}: $k \neq 2a+1$. \\
To prove this, we assume that $ k = 2a+1$ is a possibility, in which case, we have the following grouping of weight matrices.\\
%\begin{figure}
\begin{center}
\begin{tabular}{l|l|l|l|l}
$I_n$  & $F_1$ & $F_2$ & $\cdots$ & $F_{2a}$ \\ \hline
$A_{1Q}$  & $A_{2Q}$ & $A_{3Q}$ & $\cdots$ & $A_{(2a+1)Q}$ \\
\end{tabular}
\end{center}
%\label{fig_3}
%\end{figure}
Considering $A_{2Q}$, each of the terms that $A_{2Q}$ is a linear combination of should anticommute with $F_2$, $F_3$, $\cdots$ and $F_{2a}$. The only possibilities from the set $\{F_1^{\lambda_1}F_2^{\lambda_2}\cdots F_{2a}^{\lambda_{2a}}, \lambda_i \in \{0,1\}, i = 1, \cdots, 2a \}$ are $F_1$ and $F_1F_2\cdots F_{2a} = cF_{2a+1}, c= \pm j$ or $\pm 1$. Therefore, $A_{2Q} = a_{2,1}F_1+a_{2,2}F_1F_2\cdots F_{2a}$. Next, considering $A_{3Q}$, the only elements anticommuting with $F_1$, $F_3$, $\cdots$ and $F_{2a}$ are $F_2$ and $F_1F_2\cdots F_{2a}$. Therefore, $A_{3Q} = a_{3,1}F_2+a_{3,2}F_1F_2\cdots F_{2a}$. Since $A_{2Q}$ should also anticommute with $A_{3Q}$, either $a_{2,2} = 0$ or $a_{3,2} = 0$. So, either $A_{2Q} = \pm F_1$ and $A_{3Q} = a_{3,1}F_2+a_{3,2}F_1F_2\cdots F_{2a}$ or $A_{2Q} = a_{2,1}F_1+a_{2,2}F_1F_2\cdots F_{2a}$ and $A_{3Q} = \pm F_2$. In either case, the assignment violates the rule that the weight matrices are linearly independent over $\mathbb{R}$. As a result, we can't have any valid elements as the weight matrices and hence, $k \neq 2a+1$.\\

\noindent \textit{Claim 3}: $ k = 2a $. \\
Consider the following grouping scheme of weight matrices.\\
{\footnotesize
\begin{center}
\begin{tabular}{l|c|c|c|c}
$I_n$ & $\cdots$ & $F_l$  & $\cdots$ & $F_{2a-1}$ \\ \hline
$m\prod_{i=1}^{2a-1}F_i$ & $\cdots$ & $m\prod_{i=1,i\neq l}^{2a-1}F_i$ & $\cdots$ & $m\prod_{i=1}^{2a-2}F_i$ \\
\end{tabular}
\end{center}
}

In the above grouping scheme, $m = j$ if $a$ is odd, and $m=1$ if $a$ is even. It can be noted that $A_{iQ} = -A_{1Q}A_{iI}$, $i=2,3,\cdots,2a$. Clearly, the weight matrices are linearly independent over $\mathbb{R}$ and satisfy \eqref{B11}-\eqref{B17}. Hence, an SSD code transmitting $2a$ complex symbols in $2^a$ channel uses exists. This completes the proof.
\end{proof}

We observe that for $2$ transmit antennas, $k \neq 3$. So, the rate of a unitary-weight SSD code for 2 transmit antennas can be at most 1 complex symbol per channel use, which is also the rate of the well-known Alamouti code, which is single-real-symbol decodable and offers full diversity for all complex constellations. So, the unitary-weight SSD code for 2 transmit antennas offers no advantage compared to the Alamouti code. So, in the subsequent analysis, we only consider codes for $2^a$ transmit antennas, $a > 1$. \\

\begin{theorem}
 Any maximal rate, normalized unitary-weight SSD code must satisfy the following in addition to satisfying \eqref{B11}-\eqref{B17}.
\begin{eqnarray*}
A_{1Q} &=& A_{1Q}^H (equivalently, A_{1Q}^2 = I_n), \\
\label{C12}
A_{iI}A_{1Q} &=& A_{1Q}A_{iI}, \\
\label{C13}
A_{iQ} &=& \pm A_{iI}A_{1Q},
\end{eqnarray*}
for $i = 2,3,\cdots, 2a$.
\end{theorem}

\begin{proof}
 Consider the following grouping of weight matrices.\\
%\begin{figure}
\begin{center}
\begin{tabular}{l|l|l|l|l}
$I_n$  & $F_1$ & $F_2$ & $\cdots$ & $F_{2a-1}$ \\ \hline
$A_{1Q}$  & $A_{2Q}$ & $A_{3Q}$ & $\cdots$ & $A_{(2a)Q}$ \\
\end{tabular}
\end{center}
%\label{fig_4}
%\end{figure}

\begin{figure*}
\begin{center}
\begin{tabular}{l|l|l|l|l}
$A_{1I}^\prime = I_n$ & $A_{2I}^\prime =-F_1$ & $A_{3I}^\prime = -F_1F_2$ & $\cdots$ & $A_{(2a)}^\prime = -F_1F_{2a-1}$ \\ \hline
$A_{1Q}^\prime = \pm m \prod_{i=1}^{2a-1}F_i$ & $A_{2Q}^\prime = -F_1A_{1Q}$ & $A_{3Q}^\prime = \pm m\prod_{i=3}^{2a-1}F_i$ & $\cdots$ & $A_{(2a)Q}^\prime =\pm m \prod_{i=2}^{2a-2}F_i$ \\
\end{tabular}
\end{center}
\hrule
\end{figure*}

$A_{1Q}$ can have two possibilities. Either $A_{1Q}^2 = -I_n$ or  $A_{1Q}^2 \neq -I_n$.
\begin{enumerate}
\item Let $A_{1Q}^2 = -I_n$. We prove that this is not a possibility. If it were true, i.e., $A_{1Q}^2 = -I_n$, then $A_{1Q}$ should anticommute with $F_1$,  $F_2$, $\cdots$ and $F_{2a-1}$, as also seen in \eqref{B13}. The only matrices from the set $\{F_1^{\lambda_1}F_2^{\lambda_2}\cdots F_{2a}^{\lambda_{2a}}, \lambda_i \in \{0,1\}, i = 1, \cdots, 2a \}$ that anticommute with $F_1$, $F_2$, $\cdots$ and $F_{2a-1}$ are $F_{2a}$ and $F_1F_2\cdots F_{2a}$. So, let  $A_{1Q} = a_{1,1}F_{2a}+a_{1,2}F_1F_2\cdots F_{2a}$, with
\begin{equation} \label{incomp}
a_{1,1}^2+ca_{1,2}^2 =1,
\end{equation}
\noindent where $c=1$ if $(F_1F_2\cdots F_{2a})^2 = -I_n$ ($a$ is odd) and $c=-1$ if $(F_1F_2\cdots F_{2a})^2 = I_n$ ($a$ is even). Next, considering $A_{2Q}$, the only matrices from the set $\{F_1^{\lambda_1}F_2^{\lambda_2}\cdots F_{2a}^{\lambda_{2a}}, \lambda_i \in \{0,1\}, i = 1, \cdots, 2a \}$ that anticommute with $F_2$, $F_3$, $\cdots$ and $F_{2a-1}$ are $F_1$, $F_{2a}$, $F_1F_2\cdots F_{2a}$ and $F_2F_3\cdots F_{2a-1}$. As a result, $A_{2Q} = a_{2,1}F_1+a_{2,2}F_{2a}+a_{2,3}F_2F_3\cdots F_{2a-1} + a_{2,4}F_1F_2\cdots F_{2a}$. Further, since $A_{2Q}^2 = -I_n$, we have
$A_{2Q}^2 = -(a_{2,1}^2+a_{2,2}^2 - ca_{2,3}^2 +ca_{2,4}^2)I_n + 2a_{2,1}a_{2,3}F_1F_2\cdots F_{2a-1} + 2a_{2,2}a_{2,3}F_2F_3\cdots F_{2a} + 2ca_{2,3}a_{2,4}F_1F_{2a}$, with $c$ as mentioned before. Since  $I_n$, $F_2F_3\cdots F_{2a}$, $F_1F_2\cdots F_{2a-1}$ and $F_1F_{2a}$ are linearly independent over $\mathbb{C}$, either $a_{2,1} = a_{2,2} = a_{2,4} = 0$ or $a_{2,3} = 0$. Suppose $a_{2,1} = a_{2,2} = a_{2,4} = 0$, $A_{2Q} = a_{2,3}F_2F_3\cdots F_{2a-1}$. Since $A_{1Q}$ anticommutes with $A_{2Q}$, we see that $A_{2Q}$ cannot be $a_{2,3}F_2F_3\cdots F_{2a-1}$. This is because both $F_{2a}$ and $F_1F_2\cdots F_{2a}$ commute with $F_2F_3\cdots F_{2a-1}$. Therefore, $A_{2Q} = a_{2,1}F_1+A_{2,2}F_{2a}+a_{2,4}F_1F_2\cdots F_{2a}$. By a similar argument, $A_{3Q} = a_{3,1}F_2+a_{3,2}F_{2a}+a_{3,4}F_1F_2\cdots F_{2a}$.

Considering that $A_{1Q}$, $A_{2Q}$ and $A_{3Q}$ anticommute pairwise, we must have

\begin{eqnarray}
\label{eq1}
 a_{1,1}a_{2,2}+ ca_{1,2}a_{2,4}& = &0,\\
\label{eq2}
a_{1,1}a_{3,2}+ ca_{1,2}a_{3,4}& =& 0,\\
\label{eq3}
a_{2,2}a_{3,2}+ ca_{2,4}a_{3,4}& =& 0.
\end{eqnarray}
From \eqref{eq1}, \eqref{eq2} and \eqref{eq3}, we obtain
\begin{equation*}
\frac{a_{1,1}}{a_{1,2}} = \frac{-ca_{2,4}}{a_{2,2}} = \frac{-ca_{3,4}}{a_{3,2}} = \frac{a_{2,2}}{a_{2,4}}.
\end{equation*}

Hence, $a_{1,1}^2+ca_{1,2}^2 = 0$, which contradicts \eqref{incomp}. So, $A_{1Q}^2 \neq -I_n$ and $A_{1Q}$ cannot be anti-Hermitian.
\item Let $A_{1Q}^2 \neq -I_n$. In this case, we first look for possibilities for $A_{iQ}$, $i=2,\cdots, 2a$. As argued before, either $A_{2Q} = a_{2,1}F_1+a_{2,2}F_{2a}+a_{2,4}F_1F_2\cdots F_{2a}$ or $A_{2Q} = mF_2F_3\cdots F_{2a-1}$, with $m = \pm 1$ if $a$ is even and $m = \pm j$ if $a$ is odd. Assuming that $A_{2Q} = a_{2,1}F_1+a_{2,2}F_{2a}+a_{2,4}F_1F_2\cdots F_{2a}$, we have
\begin{center}
 $A_{3Q} = a_{3,1}F_2+a_{3,2}F_{2a}+a_{3,4}F_1F_2\cdots F_{2a}$,\\
 $A_{4Q} = a_{4,1}F_3+a_{4,2}F_{2a}+a_{4,4}F_1F_2\cdots F_{2a}$.\\
\end{center}
 Since $A_{2Q}^2 = A_{3Q}^2 = A_{4Q}^2 = -I_n$,
\begin{eqnarray*}
a_{2,1}^2 + a_{2,2}^2 + ca_{2,4}^2 &=& 1, \\
a_{3,1}^2 + a_{3,2}^2 + ca_{3,4}^2 &=& 1,
\end{eqnarray*}
\begin{equation}\label{deq3}
a_{4,1}^2 + a_{4,2}^2 + ca_{4,4}^2 = 1,
\end{equation}
with
\begin{equation*}
c = \left\{ \begin{array}{ccc}
 1 & \textrm{if} & (F_1F_2\cdots F_{2a})^2 = -I_n( \textrm{$a$ is odd}),\\
-1 & \textrm{if} & (F_1F_2\cdots F_{2a})^2 = I_n( \textrm{$a$ is even}).\\
\end{array}\right.
\end{equation*}
Further, considering that $A_{2Q}$, $A_{3Q}$ and $A_{4Q}$ anticommute with each other, we have the following equalities.
\begin{eqnarray*}
a_{2,2}a_{3,2} + ca_{2,4}a_{3,4}& = & 0,\\
a_{2,2}a_{4,2} + ca_{2,4}a_{4,4}& = & 0,\\
a_{3,2}a_{4,2} + ca_{3,4}a_{4,4}& = & 0.
\end{eqnarray*}
with $c$ as mentioned before. From the above set of equations, we obtain
\begin{equation*}
\frac{a_{2,2}}{a_{2,4}} = \frac{-ca_{3,4}}{a_{3,2}} = \frac{-ca_{4,4}}{a_{4,2}} = \frac{a_{4,2}}{a_{4,4}}.
\end{equation*}

So, $a_{4,2}^2+ca_{4,4}^2 = 0 \Leftrightarrow pa_{4,2} = a_{4,4}$, with $p =\pm j$ if $c = 1$ and $p = \pm 1$ if $c = -1$. So, from \eqref{deq3}, we obtain,
$a_{4,1} =  \pm 1$. By a Similar argument, we obtain, $a_{2,1} =  \pm 1,  pa_{2,2} =  a_{2,4}$, $a_{3,1} =  \pm 1,  pa_{3,2} = a_{3,4}$. Therefore,
\begin{eqnarray*}
A_{2Q}& = & \pm F_1+ a_{2,2}( F_{2a} + pF_1F_2\cdots F_{2a}), \\
A_{3Q}& = & \pm F_2+ a_{3,2}( F_{2a} + pF_1F_2\cdots F_{2a}), \\
A_{4Q}& = & \pm F_3+ a_{4,2}( F_{2a} + pF_1F_2\cdots F_{2a}).
\end{eqnarray*}
It is easy to see that the above assignment of matrices violates the linear independence of the matrices $A_{iI}, A_{iQ}$, $i = 2,3,4$ over $\mathbb{R}$. Therefore, the assumption that $A_{2Q} = a_{2,1}F_1+A_{2,2}F_{2a}+a_{2,4}F_1F_2\cdots F_{2a}$ is not valid. So, let
\begin{equation*}
A_{2Q} = \pm mF_2F_3\cdots F_{2a-1} = \pm m\prod_{i=1,i \neq 2}^{2a}A_{iI},
\end{equation*}
\noindent where,
\begin{equation*}
m = \left\{ \begin{array}{ccc}
j & \textrm{if} & \textrm{$a$ is odd},\\
1 & \textrm{if} & \textrm{$a$ is even}.\\
\end{array}\right.
\end{equation*}

Now, the only possibility is that $A_{3Q} = \pm mF_1F_3\cdots F_{2a-1} = \pm m\prod_{i=1,i\neq 3}^{2a}A_{iI}$. Similarly, by assigning $\pm m\prod_{i=1,i\neq j}^{2a}A_{iI}$ to $A_{jQ}, j=4, \cdots, 2a$, we see that the conditions in \eqref{B15}, \eqref{B16} and \eqref{B17} are satisfied and from the discussion made above, this is the only assignment possible. Now, we only need to find a valid assignment for $A_{1Q}$. Firstly, we note that
\begin{equation}
\label{mat_eq}
A_{iI}A_{iQ} = \pm A_{jI}A_{jQ}, 2 \leq i \neq j \leq 2a.
\end{equation}
From Lemma \ref{unitary_mult}, multiplying all the weight matrices by $-A_{2I}$ (i.e., $-F_1$) will result in another unitary-weight SSD code with the weight matrices grouped as shown at the top of the page, after interchanging the first and the second columns. It should be noted in the above grouping scheme that the elements in the first row except $I_n$ are all mutually anticommuting matrices and all of them also anticommute with $F_1F_{2a}$. So, $-F_1$, $-F_1F_2$, $-F_1F_3,\cdots$, $-F_1F_{2a-1}$ and $-F_1F_{2a}$ are $2a$ pairwise anticommuting matrices. Hence, instead of $F_1$, $F_2,\cdots$, $F_{2a-1}$, if we were to chose $-F_1$, $-F_1F_2$, $-F_1F_3$, $\cdots$ and $-F_1F_{2a-1}$ as the $2a-1$ anticommuting matrices, we would end up with the weight matrices as shown in the table at the top of the page. So,
\begin{equation}
\label{mat_eq1}
A_{iI}^\prime A_{iQ}^\prime = \pm A_{jI}^\prime A_{jQ}^\prime , 2 \leq i \neq j \leq 2a.
\end{equation}

\begin{figure*}
\begin{center}
\begin{equation}\label{4_Tx}
 S = \left[\begin{array}{rrrr}
        x_{1I}+jx_{2I} & x_{3I}-jx_{4Q} & x_{4I}+jx_{3Q} & x_{2Q}-jx_{1Q} \\
        -x_{3I}-jx_{4Q} & x_{1I}-jx_{2I} & x_{2Q}+jx_{1Q} & -x_{4I}+jx_{3Q} \\
        -x_{4I}+jx_{3Q} & -x_{2Q}-jx_{1Q} & x_{1I}-jx_{2I} & x_{3I}+jx_{4Q} \\
        -x_{2Q}+jx_{1Q} & x_{4I}+jx_{3Q} & -x_{3I}+jx_{4Q} & x_{1I}+jx_{2I} \\
       \end{array}\right].
\end{equation}
\end{center}
\hrule
\end{figure*}
From \eqref{mat_eq} and \eqref{mat_eq1},  $A_{1Q} = \pm mF_1F_2\cdots F_{2a-1}$. This further implies that $A_{1Q}$ must be a unitary, Hermitian matrix, because of the choice of $m$.
\end{enumerate}

\noindent So, the weight matrices of the normalized unitary-weight SSD code for $2^a$ transmit antennas are \\
%%%%%%%%%%%%%%%%%%%%%%%%%%%%%%%%%%%%%%%%%%%%%%%%%%%%%%%%%%%%%%%%%%%%%%%%%%%%%%%
%\begin{figure*}
{\footnotesize
\begin{center}
\begin{tabular}{l|l|l|l}
$I_n$ & $F_1$ &  $\cdots$ & $F_{2a-1}$ \\ \hline
$\pm m\prod_{i=1}^{2a-1}F_i$ & $\pm m\prod_{i=2}^{2a-1}F_i$ & $\cdots$ & $\pm m\prod_{i=1}^{2a-2}F_i$ \\
\end{tabular}
\end{center}
}

%\hrule
%\end{figure*}
%%%%%%%%%%%%%%%%%%%%%%%%%%%%%%%%%%%%%%%%%%%%%%%%%%%%%%%%%%%%%%%%%%%%%%%%%%%%%%%
This completes the proof of the theorem.
\end{proof}

For 4 transmit antennas, by applying the procedure outlined in \cite{TiH}, we obtain the following pairwise anticommuting, anti-Hermitian matrices.
\begin{equation*}
F_1 = \left[ \begin{array}{rrrr}
j & 0 & 0 & 0 \\
0 & -j & 0 & 0 \\
0 & 0 & -j & 0 \\
0 & 0 & 0 & j \\
\end{array}\right], ~~~~ F_2 = \left[ \begin{array}{rrrr}
0 & 1 & 0 & 0 \\
-1 & 0 & 0 & 0 \\
0 & 0 & 0 & 1 \\
0 & 0 & -1 & 0 \\
\end{array}\right]
\end{equation*}
\begin{equation*}
F_3 = \left[ \begin{array}{rrrr}
0 & 0 & 1 & 0 \\
0 & 0 & 0 & -1 \\
-1 & 0 & 0 & 0 \\
0 & 1 & 0 & 0 \\
\end{array}\right], ~~~~ F_4 = \left[ \begin{array}{rrrr}
0 & j & 0 & 0 \\
j & 0 & 0 & 0 \\
0 & 0 & 0 & j \\
0 & 0 & j & 0 \\
\end{array}\right].
\end{equation*}

 For constructing a maximal rate, unitary weight SSD code for 4 transmit antennas, we only need 3 pairwise anticommuting, anti-Hermitian matrices. Hence, choosing $F_1$, $F_2$ and $F_3$ and applying the construction method described above, we obtain a maximal rate, unitary-weight SSD code for 4 transmit antennas, a codeword matrix $S$ of which is shown in \eqref{4_Tx}, at the top of the next page.

In general, for $2^a$ transmit antennas,  we need $2a$ unitary, anti-Hermitian, pairwise anticommuting matrices. If there are exactly $2a-1$ pairwise anticommuting matrices of size $2^a \times 2^a$, then, any matrix among them is a scaled product of the other $2a-2$. So, the following observations can be made about any maximal rate, normalized unitary-weight SSD code.
\begin{enumerate}
\item Either $A_{iI},i=2,3,\cdots,2a$ are $2a-1$ among $2a+1$ pairwise anticommuting matrices and $A_{iQ}, i=2,3,\cdots,2a$ are exactly $2a-1$ pairwise anticommuting matrices, or $A_{iQ},i=2,3,\cdots,2a$ are $2a-1$ among $2a+1$ pairwise anticommuting matrices and $A_{iI}, i=2,3,\cdots,2a$ are exactly $2a-1$ pairwise anticommuting matrices.
\item $A_{1Q}$ is a Hermitian matrix and $A_{1Q} = \prod_{i=2}^{2a}A_{iI}$ if $A_{iI},i=2,3,\cdots, 2a$ are $2a-1$ among $2a+1$ pairwise anticommuting matrices, or $A_{1Q} = \prod_{i=2}^{2a}A_{iQ}$ if $A_{iQ},i=2,3,\cdots, 2a$ are $2a-1$ among $2a+1$ pairwise anticommuting matrices.
\end{enumerate}

%%%%%%%%%%%%% End of Section IV %%%%%%%%%%%%%%%%%%%%
%%%%%%% Section V begins %%%%%%%%%%%%%%%%%%%%%%%%%%%%%%
\section{Diversity and Coding gain of unitary-weight SSD codes}
\label{sec5}
We have seen in Lemma \ref{unitary_mult} that the coding gain of a unitary-weight SSD codes does not change when normalized. In this section, we obtain a common expression for the coding gain of all unitary-weight SSD codes and identify the conditions on QAM constellations that will allow unitary weight SSD codes to have full transmit diversity and high coding gain. Let $S$ and $S^\prime$ be two distinct codewords of any normalized unitary-weight SSD code $\mathcal{S}_N$. Let
\begin{eqnarray*}
S & = & \sum_{i=1}^{2a}x_{iI}A_{iI}+x_{iQ}A_{iQ},\\
S^\prime & = & \sum_{i=1}^{2a}x_{iI}^\prime A_{iI}+x_{iQ}^\prime A_{iQ},
\end{eqnarray*}
with
\begin{center}
$A_{1I} =I_n, ~~ A_{iI}A_{iQ} = \pm A_{jI}A_{jQ} = \pm A_{1Q},~~2 \leq i \neq j \leq 2a$,\\
$A_{1Q}^H = A_{1Q}, A_{iI}^H = -A_{iI}, A_{iQ}^H = -A_{iQ}, ~~~ 2 \leq i \leq 2a$.
\end{center}
Let $\Delta S \triangleq S - S^\prime$, $\Delta x_i \triangleq x_i - x_i^\prime$, $\Delta x_{iI} \triangleq x_{iI} -x_{iI}^\prime$ and $\Delta x_{iQ} \triangleq x_{iQ} -x_{iQ}^\prime$. Then,
{\footnotesize
\begin{eqnarray*}
(\Delta S)^H\Delta S & = & \left(\sum_{i=1}^{2a}\Delta x_{iI} A_{iI} + \Delta x_{iQ} A_{iQ}\right)^H \times\\
& & \left(\sum_{m=1}^{2a}\Delta x_{mI} A_{mI} + \Delta x_{mQ} A_{mQ}\right)\\
& = & \sum_{i=1}^{2a}\left(\Delta x_{iI}^2+ \Delta x_{iQ}^2\right)I_n \pm 2 \Delta x_{iI}\Delta x_{iQ}A_{iI}A_{iQ}\\
%& = & \sum_{i=1}^{2a}(\Delta x_{iI}^2+ \Delta x_{iQ}^2)I_n + 2\sum_{i=1}^{2a}\pm\Delta x_{iI}\Delta x_{iQ}A_{1Q}\\
& = & \sum_{i=1}^{2a}\left(\Delta x_{iI}^2+ \Delta x_{iQ}^2\right)I_n \pm 2\Delta x_{iI}\Delta x_{iQ}A_{1Q}.
\end{eqnarray*}
}

Since $A_{1Q}$ is unitary and Hermitian, the eigenvalues of $A_{1Q}$ are $\pm 1$ and $A_{1Q}$ is unitarily diagonalizable. Let $A_{1Q} = E\Lambda E^H$, where $E$ is unitary and $\Lambda$ is a diagonal matrix with the diagonal entries being $\pm 1$. Therefore,

{\footnotesize
\begin{eqnarray*}
(\Delta S)^H \Delta S & = & \sum_{i=1}^{2a}\left(\Delta x_{iI}^2+ \Delta x_{iQ}^2\right)EE^H \pm 2\Delta x_{iI}\Delta x_{iQ}E\Lambda E^H\\
& = & E\left(\sum_{i=1}^{2a}\left(\Delta x_{iI}^2+ \Delta x_{iQ}^2\right)I_n \pm 2\Delta x_{iI}\Delta x_{iQ}\Lambda \right)E^H
\end{eqnarray*}
}

and $det\big((\Delta S)^H\Delta S\big)$
{\footnotesize
\begin{eqnarray*}
& = & det\left(\sum_{i=1}^{2a}\left(\Delta x_{iI}^2+ \Delta x_{iQ}^2\right)I_n \pm 2\Delta x_{iI}\Delta x_{iQ}\Lambda \right)\\
& =& \prod_{j=1}^{n}\sum_{i=1}^{2a}\left(\Delta x_{iI}^2 + \Delta x_{iQ}^2+(-1)^{k_i+s_j}2\Delta x_{iI}\Delta x_{iQ}\right)
\end{eqnarray*}
\begin{eqnarray*}
& =& \prod_{j=1}^{n}\sum_{i=1}^{2a}\left(\Delta x_{iI}+(-1)^{k_i+s_j}\Delta x_{iQ}\right)^2,
\end{eqnarray*}
}

\noindent where,
\begin{equation*}
s_j = \left\{ \begin{array}{ll}
0 & \textrm{if the $(j,j)^{th}$ entry of $\Lambda$ is 1}, \\
1 & \textrm{if the $(j,j)^{th}$ entry of $\Lambda$ is -1}, \\
\end{array}\right.
\end{equation*}
 and
\begin{equation*}
k_i = \left\{\begin{array}{lll}
0 & \textrm{if} & A_{iI}A_{iQ} = A_{1Q}, \\
1 & \textrm{if} & A_{iI}A_{iQ} = -A_{1Q}. \\
\end{array}\right.
\end{equation*}
The minimum of the determinant, denoted by $\Delta_{min}$, of $(\Delta S)^H \Delta S$ for all possible non-zero $\Delta S$ is given as 
\begin{equation*}
\Delta_{min} = \min_{\Delta S \neq 0} \left(\prod_{j=1}^{n}\sum_{i=1}^{2a}\left(\Delta x_{iI}+(-1)^{k_i+s_j}\Delta x_{iQ}\right)^2\right).
\end{equation*}
Since the expression inside the bracket in the above equation is a product of the sum of squares of real numbers, its minimum occurs when all but one among $\Delta x_i, i = 1,2,\cdots 2a$ are zeros. So,
\begin{equation}\label{dmin1}
\Delta_{min} = \min_{\Delta x_i \neq 0} \prod_{j=1}^{n}\left(\Delta x_{iI}+(-1)^{k_i+s_j}\Delta x_{iQ}\right)^2.
\end{equation}

\begin{table*}
\centering
\begin{tabular}[c]{|c|c|c|c|}
\hline
Constellation: & $4$-QAM &  $16$-QAM & $64$-QAM  \\
\hline
CIOD & 10.24 & 10.24 & 10.24  \\ \hline
 MDCQOD \cite{WWX} & $10.24$ & $10.24$ & $10.24$  \\
\hline
MDCQOD \cite{YGT} & $10.24$ & $10.24$ & $10.24$  \\ \hline
New Design & $10.24$ & $10.24$ & $10.24$  \\
\hline
\end{tabular}
\caption{Comparison of the Minimum Determinants of a few SSD codes for 4 Transmit antennas}
\label{table1}
\hrule
\end{table*}

Let $m$ be the algebraic multiplicity of 1 as the eigenvalue of $A_{1Q}$ and $n-m$ be that of -1. We make use of the following lemma to conclude that $m=n-m$.
\begin{lemma}
Let $F_1$, $F_2$, $\cdots$ and $F_{2a}$ be $2^a \times 2^a$ unitary, pairwise anticommuting matrices. Then, the product matrix $F_1^{\lambda_1}F_2^{\lambda_2}\cdots F_{2a}^{\lambda_{2a}}, \lambda_i \in \{0,1\}, i=1,2,\cdots 2a$, with the exception of $I_{2^a}$, is traceless.
\end{lemma}
\begin{proof}
It is well known that $tr(AB)=tr(BA)$ for any two matrices $A$ and $B$. Let $A$ and $B$ be two invertible, $n \times n$ anticommuting matrices. So,
\begin{eqnarray*}
AB &=& -BA. \\
ABA^{-1} &=& -B.\\
tr(ABA^{-1}) &=& -tr(B).\\
tr(A^{-1}AB) = -tr(B) &\Leftrightarrow& tr(B) = -tr(B).
\end{eqnarray*}
\begin{equation}\label{trace}
\therefore tr(B) = 0.
\end{equation}
Similarly, it can be shown that $tr(A) = 0$. By applying Lemma \ref{lemma_4}, it can be seen that any product matrix $F_1^{\lambda_1^\prime}F_2^{\lambda_2^\prime}\cdots F_{2a}^{\lambda_{2a}^\prime}$, anticommutes with some other product matrix from the set $\{F_1^{\lambda_1}F_2^{\lambda_2}\cdots F_{2a}^{\lambda_{2a}}, \lambda_i \in \{0,1\}, i = 1,2,3,\cdots,2a\}$. Therefore, from the result obtained in \eqref{trace}, we can say that every product matrix $F_1^{\lambda_1}F_2^{\lambda_2}\cdots F_{2a}^{\lambda_{2a}}$ except $I_{2^a}$ is traceless.
\end{proof}

Since $A_{1Q}$ is a scaled product of $2a-1$ matrices among $2a$ unitary, pairwise anticommuting matrices, $A_{1Q}$ is traceless. Hence, $m=n-m$. So, \eqref{dmin1} becomes
\begin{equation}\label{dmin2}
\Delta_{min} = \min_{\Delta x_i \neq 0} \left(\Delta x_{iI}^2 - \Delta x_{iQ}^2\right)^n.
\end{equation}
From the above expression, it is clear that for maximal-rate, unitary-weight SSD codes to offer full transmit diversity, the difference set $\Delta \mathcal{A} \triangleq \{a-b |a,b \in \mathcal{A} \}$, where $\mathcal{A}$ is the constellation employed, should not have any points that lie on lines that are at $\pm 45$ degrees in the complex plane from the origin. Further, since the analysis leading up to the expression in \eqref{dmin2} is not specific to any particular unitary-weight SSD code, we can infer that for any particular constellation $\mathcal{A}$, all maximal-rate, unitary-weight SSD codes have the same coding gain.

%%%%%% Added later %%%%%%%%%%%%%%%%%%%%%%%%%%%%%%%%%%%%%%%
\subsection{Diversity, coding gain calculations for QAM}\label{sec5a}
In this subsection we show that all maximal-rate, unitary-weight SSD codes have the same coding gain as the CIODs \cite{KhR} for QAM constellations. Let $y_i$, $i = 1,2,\cdots,2a$ be the information symbols that take values from a constellation $\mathcal{A}_1$. Consider the following unitary rotation.
\begin{equation*}\left[\begin{array}{c}
x_{iI}\\
x_{iQ}
\end{array}\right]
= \left[\begin{array}{cc}
\frac{1}{\sqrt{2}} & -\frac{1}{\sqrt{2}} \\
\frac{1}{\sqrt{2}} & \frac{1}{\sqrt{2}}\\
\end{array}\right]\left[\begin{array}{c}
y_{iI}\\
y_{iQ}
\end{array}\right], \forall i = 1,2,\cdots,2a.
\end{equation*}
The above operation is equivalent to rotating $y_i$ by $\pi/4$ radians to obtain $x_i$.
Then, from \eqref{dmin2}, we have
\begin{equation}\label{rot_sym}
\Delta_{min} =  \min_{\Delta y_i \neq 0}(2\Delta y_{iI} \Delta y_{iQ})^n.
\end{equation}
The above expression is the same as the one for CIOD, obtained in \cite{KhR}. It is to be noted that \eqref{rot_sym} holds even when the angle of rotation is $-\pi/4$ radians. In order to maximize $\Delta_{min}$, the minimum of the product $\vert \Delta y_{iI} \Delta y_{iQ} \vert$, called the \textit{product distance}, must be maximized. This has been done for QAM in \cite{KhR}, by rotating QAM constellations by an angle of $\pm \frac{1}{2}tan^{-1}2$. So, $y_i$, $i=1,2,\cdots,2a$, should take values from a rotated QAM constellation, with the angle of rotation being $\pm \frac{1}{2}tan^{-1}2$. So, the original information symbols $x_i$, $i=1,2,\cdots,2a$ should take values from a rotated QAM constellation, the angle of rotation being $\pm \frac{\pi}{4} \pm \frac{1}{2}tan^{-1}2$. Since the coding gain for CIOD has been maximized in \cite{KhR} by using a $\pm \frac{1}{2}tan^{-1}2$ radian rotated QAM constellation, the coding gain for all unitary-weight SSD codes when the symbols take values from a $\pm \frac{\pi}{4} \pm \frac{1}{2}tan^{-1}2$ radian rotated QAM constellation is also maximized.  Table \ref{table1} gives a comparison of the minimum determinants for the CIOD, MDCQOD and the unitary-weight SSD code presented in \eqref{4_Tx}, all the codes designed for 4 transmit antennas. In the calculations, all the codes have the same average energy but the constellation energy has been allowed to increase with the increase in constellation size. As analytically shown, the minimum determinants are the same for all the three codes.

\begin{figure}
\centering
\includegraphics[width=3.4in,height = 3.4in]{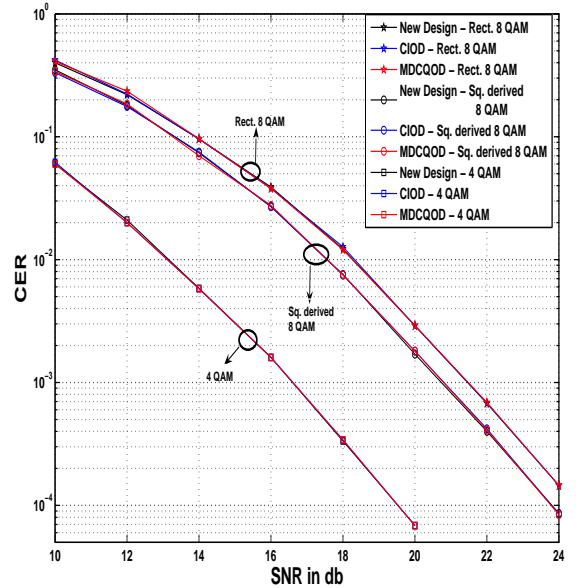}
\caption{Comparison of the CER performance of the SSD codes and the CIOD}
\label{plot}
\end{figure}

\subsection{Simulation results for our SSD codes}
\label{sec5b}
In this subsection, we provide some simulation results for 4 transmit antennas. All simulations are done assuming a quasi-static Rayleigh fading channel. The number of receive antennas is 1. Fig. \ref{plot} shows the codeword error rate (CER) performances of the CIOD for 4 transmit antennas, the MDCQOD for 4 transmit antennas \cite{WWX}, and the new design whose codeword matrix is as in \eqref{4_Tx}, at 2 bits and 3 bits per channel use (bpcu). For transmission at 2 bpcu, the constellation employed is the $\frac{1}{2}tan^{-1}2$ radian rotated 4-QAM for the CIOD and the $\frac{\pi}{4}+\frac{1}{2}tan^{-1}2$ radian rotated 4-QAM for the MDCQOD and the new design. For transmission at 3 bpcu, the constellations employed are the rotated rectangular 8-QAM and the rotated square-derived 8-QAM, the angle of rotation being the same as in the case of 4-QAM. A squared-derived 8-QAM constellation is obtained by removing the signal point with the highest energy from a 9-QAM. Specifically, it is the set $\{ -1-j,-1+j,-1+3j,1-j,1+j,1+3j,3-j,3+j\}$. The simulation results support the fact that for QAM constellations, the coding gain of the SSD codes is the same as that of the CIOD.

%%%%%%%%%%%%%%%%%%%%%%%%%%%%%
%%%%%%%%%%%%% Examples end %%%New section begins %%%%%%%%%%%%%%%%%

%%%%%%%%%%%%%%%%%%%%%%%%%%%
\section{Discussion and Concluding Remarks}\label{sec6}
In this paper, we have provided an achievable upper bound on the rate of unitary-weight SSD codes for $2^a$ transmit antennas. The upper bound has been shown to be $\frac{a}{2^{a-1}}$ complex symbols per channel use. We also have completely characterized the structure of the weight matrices of the codes meeting the upper bound. We have further shown that all unitary-weight SSD codes that meet the upper bound have the same coding gain as that of the CIODs and we have also identified the angle of rotation for QAM constellations that allow the codes to have optimum coding gain. The analysis done in this paper throws open the following questions.
\begin{enumerate}
\item What is the upper bound on the rate of square, non-unitary-weight SSD codes? Further, what are the conditions on the signal constellation that allow non-unitary-weight SSD codes to achieve full-diversity and optimum coding gain?
\item What is the upper bound on the rate of rectangular SSD codes, the class of which the rectangular MDCQODs presented in \cite{WWX} and \cite{YGT} are a subclass?
\end{enumerate}
Further, the analysis in this paper can be used to study the rates of multi-symbol decodable codes, the upper bounds of which has never been reported in literature. These questions provide some directions for future research.

%%%%%%%%%%%%%%%%%%%%%%%%%%%%%%%%%%%
\section*{Acknowledgement}
This work was partly supported by
the DRDO-IISc Program on Advanced Research in Mathematical
Engineering and by the Council of Scientific \&
Industrial Research (CSIR), India, through Research Grant (22(0365)/04/EMR-II) to B.S.~Rajan.\\
We thank X.-G.Xia for sending the preprint of \cite{WWX}. We would also like to thank the anonymous Reviewers for pointing out \cite{anti_matric}, which helped in simplifying the proof for the upper bound on the rate of unitary-weight SSD codes.

%%%%%%%%%%%%%%%%%%%%%%%%%%%%%%%%%%%

%%%%% References begin %%%%%%%%%

%%%%%%%%%%%%%%%%%%%%%%%%%%%%%%%%%%%%%%%%%%%%%%%%%%%%%%%%%
\end{document}